\documentclass[a4paper,12pt]{article}
\usepackage{amssymb,amsmath,amsthm,epsfig}

\newcommand{\RR}{{\mathbb{R}}}
\newcommand{\CC}{{\mathbb{C}}}
\newcommand{\ZZ}{{\mathbb{Z}}}
\newcommand{\pa}{\partial}
\newcommand{\ii}{{\rm i}}
\newcommand{\dd}{{\rm d}}
\newcommand{\sfrac}[2]{{\textstyle\frac{#1}{#2}}}
\newcommand{\Tr}{\mathrm{Tr}}
\newcommand{\uu}{\mathfrak{u}}
\newcommand{\su}{\mathfrak{su}}
\newcommand{\T}{\mathbf{T}}
\newcommand{\y}{\mathbf{y}}

\numberwithin{equation}{section}

\theoremstyle{plain} \newtheorem{thm}{Theorem}
\theoremstyle{plain} \newtheorem{prop}[thm]{Proposition}
\theoremstyle{plain} \newtheorem{cor}[thm]{Corollary}
\theoremstyle{definition} 
\theoremstyle{plain} \newtheorem{conj}[thm]{Conjecture}

\title{The large $N$ limit of the Nahm transform}
\author{Derek Harland\footnote{email address: d.g.harland@durham.ac.uk}
  \bigskip
  \\Department of Mathematical Sciences,
  \\Durham University,
  \\DH1 3LE}
\date{14th February 2010}

\begin{document}

\maketitle

\begin{abstract}
 We consider the large $N$ limit of the Nahm transform, which relates charge $N$ monopoles to solutions to the Nahm equation involving $N\times N$ matrices.  In the large $N$ limit the former approaches a magnetic bag, and the latter approaches a solution of the Nahm equation based on the Lie algebra of area-preserving vector fields on the 2-sphere.  We show that the Nahm transform simplifies drastically in this limit.
\end{abstract}

\section{Introduction}
\label{sec:1}

Magnetic monopoles in SU(2) Yang-Mills-Higgs theory are static solutions of the field equations in 3 dimensions.  They are classified by an integer $N$, the topological charge.  The Nahm equation is a non-linear ODE for a triple of matrix-valued functions, and solutions can again be classified by an integer, the size $N$ of the matrices.  The Nahm transform \cite{nahm,cg,hitchin} is a one-to-one correspondence between charge $N$ magnetic monopoles (in the BPS limit) and $N\times N$ matrix solutions of the Nahm equation.

The magnetic bag conjecture, originally formulated by Bolognesi \cite{bol1}, concerns the large $N$ limit magnetic monopoles.  The conjecture states that for large $N$ there are charge $N$ magnetic monopoles whose fields behave in an essentially abelian way except on a thin spherical shell.  In the limit $N\to\infty$, the shell becomes infinitely thin compared with its diameter, so the non-linear field equations become linear almost everywhere.  A more refined picture of magnetic bags was subsequently presented by Lee and Weinberg \cite{lw}.  More recently, magnetic bags have been coupled to gravity \cite{bol2} and found an interpretation in condensed matter models based on holography \cite{bt}.

The large $N$ limit of the Nahm equations was considered much earlier by Ward \cite{ward}.  The main idea is that the space of real functions on a surface can form a Lie algebra $\uu(\infty)$, with the Lie bracket given by a Poisson bracket; an idea going back to Hoppe \cite{hoppe} states that this Lie algebra should be considered in some sense to be the large $N$ limit of the Lie algebra of $\uu(N)$ of hermitian $N\times N$ matrices.  Ward showed that, somewhat surprisingly, the non-linear $\uu(\infty)$ Nahm equations are equivalent to a set of linear equations.  Many other intiguing interpretations of the $\uu(\infty)$ Nahm equation have been found by Donaldson \cite{don}, who was apparently unaware of Ward's earlier work.

Given that both sides of the Nahm transform have an $N\to\infty$ limit, and that both sides ``linearise'' in this limit, it is tempting to speculate that the Nahm transform itself persists in the large $N$ limit.  The purpose of the present article is to argue that this is indeed the case.  A Nahm transform for magnetic bags will be presented, which adds to the growing list of generalised Nahm transforms \cite{jardim}.


An outline of the remainder of this article is as follows.  We will recall in section \ref{sec:2} essential features of magnetic monopoles and magnetic bags.  Our main result is in section \ref{sec:3}: we will describe a simple transform which relates solutions of the $\uu(\infty)$ Nahm equation to magnetic bags.  Sections \ref{sec:4} and \ref{sec:5} discuss how our Nahm transform for magnetic bags can be understood as the large $N$ limit of the traditional Nahm transform, using the concept of the ``fuzzy sphere''.  Section \ref{sec:4} introduces the fuzzy sphere, while section \ref{sec:5} discusses the Nahm transform.  Section \ref{sec:5} also contains a discussion of the string theoretical picture of our Nahm transform.  We illustrate our proposal in section \ref{sec:6} with a simple example.  In section \ref{sec:7} we summarise our results and discuss some interesting open problems.

\section{Magnetic bags}
\label{sec:2}

Our starting point is SU(2) Yang-Mills-Higgs theory in the BPS limit.  Let $y^i$ be coordinates on $\RR^3$, let $A = A_i\dd y^i$ denote an SU(2) gauge field, and let $\Phi$ denote an $\mathfrak{su}(2)$-valued scalar field, where $i=1,2,3$.  The covariant derivative of $\Phi$ and the field strength are
\begin{eqnarray}
D\Phi &=& \dd\Phi + e(A\Phi-\Phi A) \\
F &=& \dd A + e A\wedge A,
\end{eqnarray}
with $e>0$ the coupling constant.  The static Yang-Mills-Higgs energy functional is
\begin{equation}
E = -\frac{1}{4}\int_{\RR^3} \Tr\left(F\wedge\ast F + D\Phi\wedge\ast D\Phi\right).
\end{equation}
By choosing appropriate units of length and energy, one could arrange for the coupling constant $e$ to be 1; however, it will prove convenient to keep $e$ as a free parameter in the theory.

We impose the boundary condition,
\begin{equation}
\|\Phi\|:=\sqrt{-\sfrac{1}{2}\Tr(\Phi^2)} \rightarrow v \mbox{ as }r:=\sqrt{y^iy^i}\rightarrow\infty,
\end{equation}
for some $v>0$.  Then the asymptotic values of $\Phi$ define a map from the boundary $S^2_\infty$ of $\RR^3$ to $S^2\subset\mathfrak{su}(2)$.  Such maps are classified by a degree $N\in\ZZ$, which is called the topological charge.  The magnetic charge is
\begin{equation}
\label{magnetic charge}
q := -\frac{1}{2}\int_{S^2_\infty} \frac{\Tr(F\Phi)}{\|\Phi\|} = \frac{2\pi N}{e}.
\end{equation}

By reversing orientations if necessary, we can assume that $N\geq0$.  Then a well-known argument due to Bogomolny shows that there is a lower bound on the energy,
\begin{equation}
\label{bog}
E \geq v q,
\end{equation}
which is saturated if and only the Bogomolny equation 
\begin{equation} \label{bog eq} F=\ast D\Phi \end{equation}
holds.  A solution of (\ref{bog eq}) is called a \emph{magnetic monopole}.  It is known that the space of gauge equivalence classes of solutions to (\ref{bog eq}) is very large: in fact, it forms a manifold of dimension $4|N|-3$.

Magnetic bags are solutions of U(1) Yang-Mills-Higgs theory with a source of magnetic charge on a closed surface $\Sigma\subset\RR^3$.  We will assume that all fields are zero on the interior of $\Sigma$, so the Higgs field and gauge field will be a real function $\phi$ and a real 1-form $a$ defined on local coordinate patches of the exterior $W$ of $\Sigma$.  The static energy functional is
\begin{equation}
\label{bag energy}
E = \frac{1}{2} \int_{W} \dd\phi\wedge\ast\dd\phi + f\wedge\ast f,
\end{equation}
where $f = \dd a$ is the field strength.  We impose the following boundary conditions on the Higgs field:
\begin{equation}
\label{bag bc}
 \phi = 0 \mbox{ on } \Sigma ,\quad \phi\to v \mbox{ as } r\to\infty.
\end{equation}
The magnetic charge of a magnetic bag is defined to be
\begin{equation}
\label{bag charge}
 q = \int_{S^2_\infty} f.
\end{equation}
It is easy to show that the magnetic bag energy admits a Bogomolny-type bound:
\begin{eqnarray}
 E &=& \frac{1}{2}\int_W |\dd\phi-\ast f|^2\,\dd^3 y + \int_W\dd\phi\wedge f \\
 &\geq& \int_{\partial W} \phi f \\
 &=& vq.
\end{eqnarray}
The bound is saturated if and only if the Bogomolny equation,
\begin{equation}
\label{bag bog}
 \dd\phi = \ast f,
\end{equation}
holds.  A solution of (\ref{bag bog}) satisfying the boundary conditions (\ref{bag bc}) is called a \emph{magnetic bag}.  In the more refined terminology of \cite{lw} this is an ``abelian bag''; we will not have anything to say about ``non-abelian bags'' in this article.  We will restrict attention to the case where $\Sigma$ is connected with genus 0, in other words, $\Sigma$ is topologically a 2-sphere.  However, there are many other possibilities: for example, $\Sigma$ could be a torus, or a disjoint union of multiple 2-spheres.

It follows from (\ref{bag bc}), (\ref{bag charge}) and (\ref{bag bog}) that, at large distances from the origin,
\begin{equation}
\label{bag asymptotics}
\phi = v - \frac{q}{4\pi r} + O(r^{-2}).
\end{equation}
The simplest example of a magnetic bag is Bolognesi's spherical bag.  It has $\Sigma$ a 2-sphere with radius
\begin{equation}
\label{bag radius}
R=q/(4\pi v),
\end{equation}
and spherically symmetric scalar field $\phi$:
\begin{equation}
\label{bag phi}
\phi(r) = \begin{cases} v\left(1-R/r\right) & R\leq r \\
           0 & 0\leq r<R.
          \end{cases}
\end{equation}
The field strength $f$ is obtained from the Bogomolny equation (\ref{bag bog}):
\begin{equation}
 f = \ast\dd\phi = \begin{cases} (q/8\pi r^3)\epsilon_{ijk}y^i\dd y^j\wedge\dd y^k & R\leq r \\
           0 & 0\leq r<R.
          \end{cases}
\end{equation}
Note that, because the magnetic charge is non-zero, the gauge potential $a$ can only be written on local coordinate patches of $W$.

Bolognesi's magnetic bag conjecture states that there are sequences of monopoles which converge to given magnetic bags as their charge $N$ tends to infinity.  From (\ref{magnetic charge}), it is clear that the coupling constant $e$ should tend to $\infty$ in order that the magnetic charge (and hence the energy) remains finite.  Formally, we will state the conjecture as follows:
\begin{conj}
\label{conj:bag}
For any magnetic bag $f,\phi,\Sigma$, there is a sequence $A^{(k)},\Phi^{(k)}$ of charge $N_k$ monopoles with coupling constant $e_k$, such that
$N_k\rightarrow\infty$, $2\pi N_k/e_k\rightarrow q$, $\|\Phi\|\rightarrow\phi$, and $-\sfrac{1}{2}\Tr(F\Phi)/\|\Phi\|\rightarrow f$ as $k\rightarrow\infty$.
\end{conj}
In physical terms, the non-zero vacuum expectation value for $\Phi$ breaks the gauge symmetry to U(1): the SU(2) gauge field decomposes into an abelian gauge field and a $W$-boson of mass $e\|\Phi\|$.  The field $W$ is close to zero except on a thin shell, and in the limit $N,e\to\infty$ the thickness of the non-abelian shell tends to 0.

There are two main pieces of evidence that support the magnetic bag conjecture.  First, the tetrahedral, octahedral, and icosahedral monopoles with charges $N=3$, 5, and 11 seem to be the first three members of a sequence converging to the spherical bag: they resemble spherical shells, with $\Phi$ close to zero inside the shell and $W$ close to zero both inside and outside the shell \cite{lw}.  Second, Ward has constructed a flat non-abelian monopole wall \cite{wall}, which serves as a prototype for the non-abelian shell of magnetic monopole with large charge.

\section{The Nahm transform for magnetic bags}
\label{sec:3}

In this section we will describe an explicit one-to-one correspondence between magnetic bags and solutions of the $\uu(\infty)$ Nahm equation.  Let $S^2$ denote the sphere with its standard rotationally-invariant area form $\omega$, normalised so that the surface area of the sphere is $4\pi$, and let $I$ denote the interval $[0,v)$ with coordinate $s$.  There is a natural Poisson bracket on $S^2$ which lifts to $S^2\times I$, and which may be defined as follows:
\begin{equation}
 \label{def pb}
 \{g,h\}\,\omega\wedge\dd s = \dd g\wedge\dd h\wedge\dd s.
\end{equation}
One finds that
\begin{equation}
 \{\hat x^i,\hat x^j\} = \epsilon_{ijk}\hat x^k,
\end{equation}
where $\hat x^1,\hat x^2,\hat x^3$ are the coordinate functions on $\RR^3$ restricted to the unit sphere.  The space of real functions on the sphere equipped with this Poisson bracket forms a Lie algebra, which will be denoted $\uu(\infty)$.  The Poisson bracket associates to each element of $\uu(\infty)$ a vector field on $S^2$, called a Hamiltonian vector field.  Hamiltonian vector fields preserve area, and in fact on $S^2$ all area-preserving vector fields are Hamiltonian.

The Poisson bracket can be used to define a Nahm equation.  Let $t^1,t^2,t^3$ be three functions on $S^2\times I$.  The $\uu(\infty)$ Nahm equation is
\begin{equation}
 \label{NEinfty}
 \frac{\dd t^i}{\dd s} = \frac{4\pi}{q}\frac{1}{2}\epsilon_{ijk} \{t^j,t^k\}.
\end{equation}
A set of $\uu(\infty)$ \emph{Nahm data} consists of a triple $t^1,t^2,t^3$ solving the Nahm equation (\ref{NEinfty}) and the boundary condition
\begin{equation}
 \label{BCinfty}
 t^i(s) = \frac{q}{4\pi} \frac{\hat x^i}{v-s} + O(1) \mbox{ as }s\rightarrow v.
\end{equation}

A set of Nahm data $t^i$ defines a map from $S^2\times I$ to $\RR^3$, which we assume to be invertible.  Thus the coordinate function $s$ and the area form $\omega$ can be pulled back to a region in $\RR^3$.  By multiplying both sides of the Nahm equation with $\omega\wedge\dd s$ and using the definition (\ref{def pb}) of the Poisson bracket, the following equation can be obtained:
\begin{equation}
 \dd t^i\wedge\omega = \frac{4\pi}{q}\frac{1}{2}\epsilon_{ijk}\, \dd t^j\wedge\dd t^k \wedge\dd s.
\end{equation}
Now the functions $t^i$ are just the coordinates $y^i$ on $\RR^3$, so this equation can be rewritten
\begin{equation}
\label{eq:3.1}
 \dd y^i \wedge\omega = \frac{4\pi}{q} \ast\dd y^i\wedge\dd s,
\end{equation}
where the Hodge star $\ast v$ of a $p$-form $v$ is defined by $u\wedge\ast v = g(u,v) \,\dd^3y$ for all $u\in\Lambda^p$.  Simple manipulations reveal that equation (\ref{eq:3.1}) is equivalent to
\begin{equation}
 \dd s = \ast \frac{q\,\omega}{4\pi}.
\end{equation}

This equation is nothing but the Bogomolny equation (\ref{bag bog}) for a magnetic bag, with $\phi=s$ and $f=q\,\omega/4\pi$.  The asymptotics (\ref{bag asymptotics}) for the field $\phi=s$ also follow directly from the boundary condition (\ref{BCinfty}) for Nahm data, and the bag surface is simply the set $\Sigma=t^j(S^2\times\{0\})$ on which $\phi=0$.  Our choice of normalisation for $\omega$ means that the magnetic charge of the bag is $q$.  Therefore a set of Nahm data determines a magnetic bag, with the coordinate $s$ playing the role of the Higgs field and the area form $\omega$ playing the role of the magnetic field.  This construction can be inverted, as the following theorem shows:

\begin{thm}
\label{thm:1}
There is a one-to-one correspondence between
\begin{itemize}
 \item sets of $\uu(\infty)$ Nahm data with the property that the map from $S^2\times I$ to $\RR^3$ defined by $t^i$ is one-to-one, and
 \item magnetic bags with the property that the bag surface $\Sigma$ is diffeomorphic to a 2-sphere and $\dd\phi\neq0$ everywhere on the exterior $W$ of $\Sigma$.
\end{itemize}
\end{thm}

\begin{proof}
We have already described the map from Nahm data to magnetic bags, so to complete the proof we need to describe the inverse construction.  There is a natural vector field $\pa/\pa s$ on $S^2\times I$ which has the following two properties:
\begin{equation}
 L_{\frac{\pa}{\pa s}} s = 1,\quad L_{\frac{\pa}{\pa s}} \omega = 0,
\end{equation}
where $L$ denotes the Lie derivative.  In order to invert the Nahm transform it is convenient to find  a vector field $Y$ on $W$ which satisfies the analogous conditions:
\begin{equation}
 L_Y \phi = 1,\quad L_Y f = 0.
\end{equation}
The unique vector field $Y$ satisfying both conditions is
\begin{equation}
 Y = \frac{\dd\phi^\sharp}{|\dd\phi|^2} = \frac{\pa_i\phi}{\pa_j\phi\pa_j\phi} \frac{\pa}{\pa y^i},
\end{equation}
Note that $Y$ is well-defined, because $\dd\phi\neq0$ by assumption.

This vector field $Y$ enables us to define a map from $S^2\times I$ to $W$.  To begin, $\Sigma$ is diffeomorphic to $S^2$ by assumption.  One can always find a map $h:S^2\to\Sigma$ such that $h^\ast f = q\,\omega/4\pi$, as follows from Moser's theorem (see for example \cite{mcduff}).  The idea is to extend this to a map $H:S^2\times I\rightarrow W$ using $Y$.  Let $p$ denote a point in $S^2$.  We can construct a flow line from $p$ using $Y$, and $H(p,s)$ is defined to be the unique point on the flow line which satisfies $\phi(H(p,s))=s$.  More explicitly, if $y(\sigma)$ is a solution of the differential equation,
\begin{eqnarray}
 \frac{\dd y^i}{\dd \sigma} &=& Y^i(y(\sigma)), \\
 y(0) &=& h(p),
\end{eqnarray}
then $H(p,s)$ is defined to be the point $y(s)$.  The point $y(s)$ lies on the level set $\{\phi=s\}$, because $\dd\phi(y(\sigma))/\dd\sigma = 1$.

Having constructed the map $H$, it remains to be shown that the coordinate functions $t^i=y^i\circ H$ satisfy the Nahm equation and boundary conditions -- this can be done by reversing the steps in the first part of the proof.  It is straightforward to check that the composition of the inverse Nahm transform described here and the Nahm transform described above is the identity.
\end{proof}

As a simple example, consider the following Nahm data \cite{ward}:
\begin{equation}
 t^i(s) = \frac{q}{4\pi} \frac{\hat x^i}{v-s}.
\end{equation}
The magnetic bag obtained via the $\uu(\infty)$ Nahm transform described in theorem \ref{thm:1} is the spherical magnetic bag described in section \ref{sec:2}.

\subsection{Linerisation of the Nahm equations}

Theorem \ref{thm:1} is essentially a correspondence between solutions of the $\uu(\infty)$ Nahm equation and harmonic functions on $\RR^3$.  The existence of such a correspondence was mentioned by Donaldson \cite{don}, but the proof given here is our own.  Previously, Ward was able to demonstrate that harmonic functions lead to solutions of the $\uu(\infty)$ Nahm equation \cite{ward}.  Theorem \ref{thm:1} fully accounts for Ward's observation that the $\uu(\infty)$ Nahm equations ``linearise''.  To make this explicit, we will briefly recall Ward's original treatment of the $\uu(\infty)$ Nahm equations.

The starting point in \cite{ward} was the observation that the Nahm equation is equivalent to $[D_0,D_1]=0$, where $D_0, D_1$ are two differential operators depending on a spectral parameter $\zeta$:
\begin{eqnarray}
D_0 &=& -\frac{4\pi}{q} \{t^1+\ii t^2,  \cdot\}  + \zeta\left(\frac{\dd }{\dd s} + \frac{4\pi\ii}{q}\{t^3, \cdot\} \right) \\
D_1 &=& \left(\frac{\dd }{\dd s} - \frac{4\pi\ii}{q}\{t^3, \cdot\}\right) + \zeta\frac{4\pi}{q} \{t^1-\ii t^2,  \cdot\}.
\end{eqnarray}
The Nahm equations imply that $D_0, D_1$ factorise:
\begin{equation}
\left( \begin{array}{c} -\ii D_1 \\ D_0 \end{array} \right) = \xi \sigma^j \left( \begin{array}{c} -\ii \\ \zeta \end{array} \right) \frac{\pa }{\pa y^j}.
\end{equation}
Here $\xi$ is the $2\times 2$ matrix,
\begin{equation}
\xi = \frac{\pa t^i}{\pa s}\sigma^i = \frac{2\pi}{q}\epsilon_{ijk}\{t^j, t^k\} \sigma^i,
\end{equation}
which agrees with the matrix introduced in \cite{ward} up to a change of basis.  Ward showed that the Nahm equation is equivalent to the following non-linear equation for $\xi$:
\begin{equation}
\sigma^k \xi \frac{\pa \xi}{\pa y^k} = 0.
\end{equation}
Moreover, by introducing $\mu=\xi/\det(\xi)$, the following linear equation was obtained:
\begin{equation}
\label{ward equation}
\sigma^j \frac{\pa \mu}{\pa y^j} = 0.
\end{equation}

Ward's linear equation (\ref{ward equation}) has a natural interpretation in terms of Clifford algebras.  The gamma matrices in 3 dimensions are just the Pauli matrices $\sigma^i$, so the operator $\sigma^j\pa_j$ is the image in the Clifford algebra of the exterior derivative $\dd$.  Likewise, the traceless hermitian matrix $\mu=\mu_j\sigma^j$ is the image in the Clifford algebra of a 1-form $\mu_j\dd y^j$.  Equation (\ref{ward equation}) says that this 1-form is both closed and co-closed:
\begin{equation}
\pa_i\mu_j-\pa_j\mu_i=0,\quad \pa_i\mu_i=0.
\end{equation}
So $\mu_j\dd y^j$ is the exterior derivative of a harmonic function, at least locally.  In fact, it is not hard to see that $\mu_j\dd y^j = \dd s$, so once again we have deduced that $s$ defines a harmonic function on $\RR^3$.

\subsection{Conserved charges}

Like the $\uu(N)$ Nahm equation, the $\uu(\infty)$ Nahm equation is integrable, and infinite tower of conserved charges can be constructed in an analogous way.  The Nahm equation is equivalent to a Lax equation,
\begin{equation}
\label{NE Lax}
\frac{\dd T}{\dd s} = \frac{4\pi}{q}\{T,-\ii t^3 + \zeta(t^1-\ii t^2)\},
\end{equation}
where
\begin{equation}
T(\zeta) := -D_0 + \zeta D_1 = \frac{4\pi}{q}\left( (t^1+\ii t^2) - 2\ii\zeta t^3 + \zeta^2(t^1-\ii t^2)\right).
\end{equation}
It follows that the charges,
\begin{equation}
C_l(\zeta,s) := \int_{S^2\times\{s\}} T^l\omega, \quad  l\in\ZZ_{>0},
\end{equation}
are conserved: $\pa_s C_l=0$.  The charge $C_l$ is a degree $2l$ polynomial in $\zeta$ satisfying a reality condition $\overline{C_l(\zeta)}=\bar{\zeta}^{2l}C_l(-\bar{\zeta}^{-1})$, so transforms in the $(2l+1)$-dimensional real irreducible represenation of SU(2).

In the magnetic bag picture, the charges $C^l$ are integrals over level sets of $\phi$:
\begin{equation}
C_l(\zeta,s) = \int_{\{\phi=s\}} T^l \ast\dd\phi.
\end{equation}
It is straightforward to check directly that these integrals are independent of $s$, by using the fact that $T^l$ (regarded as a function of $y^j$) solves the Laplace equation for all $l$.

The conserved charges associated with the spherical magnetic bag all vanish.  In this sense they provide a measure of the failure of a magnetic bag to be spherically symmetric.

\section{The fuzzy sphere}
\label{sec:4}

Thus far we have introduced a correspondence between the $\uu(\infty)$ Nahm equations and magnetic bags.  In the next section we will argue that this Nahm transform for magnetic bags can be considered to be the large $N$ limit of the Nahm transform for magnetic monopoles.  Our discussion will be based on the concept of the ``fuzzy sphere'' \cite{hoppe,madore}, and in this section we will review basic facts about fuzzy spheres.  Our approach will be quite algebraic, following the spirit of \cite{madore}, but for rigorous analytic discussions see \cite{fuzzy}.

Let $J^1,J^2,J^3$ denote the generators of the irreducible $N$-dimensional representation of $\su(2)$, satisfying $J^iJ^j-J^jJ^i=\epsilon_{ijk}J^k$, and let $X^j=(2\ii/N)J^j$.  The matrices $X^j$ satisfy the identities,
\begin{eqnarray}
\label{eq:4.1}
X^iX^j-X^jX^i &=& \frac{2\ii}{N}\epsilon_{ijk}X^k \\
\label{eq:4.2}
(X^1)^2+(X^2)^2+(X^3)^2 &=& 1 - \frac{1}{N^2},
\end{eqnarray}
(the second identity is deduced from the quadratic Casimir of the representation of $\su(2)$).  The space $\CC^{N\times N}$ of all complex $N\times N$ matrices is generated as an algebra by the $X^j$, subject to these relations.  Similarly, the space of complex functions on $S^2$ is generated as an algebra by the coordinate functions $\hat x^1,\hat x^2,\hat x^3$ on $\RR^3$ restricted to the unit sphere, subject to the relations
\begin{eqnarray}
\label{eq:4.3}
\hat x^i\hat x^j-\hat x^j\hat x^i &=& 0 \\
\label{eq:4.4}
(\hat x^1)^2 + (\hat x^2)^2 + (\hat x^3)^2 &=& 1
\end{eqnarray}
Now, as $N\to\infty$, the relations (\ref{eq:4.1}), (\ref{eq:4.2}) converge to the relations (\ref{eq:4.3}), (\ref{eq:4.4}), provided we identify $\hat x^i$ with the large $N$ limit of $X^i$.  In this sense the space of complex functions on the unit sphere can be identified with the large $N$ limit of $\CC^{N\times N}$.  The latter is commonly called the space of complex functions on the fuzzy sphere in this context.


Besides being an algebra, $\CC^{N\times N}$ is also a Lie algebra, with Lie bracket given by
\begin{equation}
\label{fuzzy PB}
\{F,G\} := \frac{N}{2\ii}(FG-GF)\quad \forall F,G\in\CC^{N\times N}.
\end{equation}
(This unusual notation and normalisation for the Lie bracket is chosen to match our conventions elsewhere.)  The Lie bracket satisfies the identities,
\begin{eqnarray}
\{ X^i,X^j \} &=& \epsilon_{ijk} X^k \\
\{F, GH \} &=& \{F,G\}H + G\{F,H\} \quad\forall F,G,H\in\CC^{N\times N},
\end{eqnarray}
while the Poisson bracket of functions on the unit sphere satisfies,
\begin{eqnarray}
\{ \hat x^i,\hat x^j \} &=& \epsilon_{ijk} \hat x^k \\
\{f, gh \} &=& \{f,g\}h + g\{f,h\} \quad\forall f,g,h:S^2\to \CC.
\end{eqnarray}
Thus the Lie bracket on $\CC^{N\times N}$ converges to the Poisson bracket on the space of complex functions on the sphere.

Earlier, we introduced Lie algebra $\uu(\infty)$ of real functions on the sphere.  The matrices $X^i$ are hermitian, so a natural definition for the space of real functions on the fuzzy sphere is the space $\uu(N)$ of hermitian matrices.  Although not closed under the multiplication, $\uu(N)$ is closed under the Lie bracket introduced above.  Thus the Lie algebra $\uu(\infty)$ is the large $N$ limit of the Lie algebra $\uu(N)$.

A final aspect to the relation between $\uu(\infty)$ and $\uu(N)$ involves the natural action of $\su(2)$ on both of them.  The matrices $X^i$ span an $\su(2)$ subalgebra of $\uu(N)$, and this $\su(2)$ acts on $\uu(N)$ via the adjoint representation.  This representation of $\su(2)$ decomposes into irreducibles as follows:
\begin{equation}
\uu(N) = N\otimes N = 1\oplus 3\oplus 5\oplus \dots \oplus 2N-3\oplus 2N-1.
\end{equation}
The irreducible representations are eigenspaces of the quadratic Casimir,
\begin{equation}
\mbox{Cas}F := \{X^1,\{X^1,F\}\} + \{X^2,\{X^2,F\}\} + \{X^3,\{X^3,F\}\}.
\end{equation}
If $F$ belongs to the $2l+1$-dimensional representation, then $\mbox{Cas} F=-l(l+1)F$.  Similarly, coordinate functions $\hat x^i$ define an $\su(2)$-subalgebra of $\uu(\infty)$, and the decomposition of this representation into irreducible representations under the adjoint action of $\su(2)$ is,
\begin{equation}
\uu(\infty) = 1\oplus 3\oplus 5\oplus \dots.
\end{equation}
The quadratic Casimir for this representation coincides with the Laplacian,
\begin{equation}
\mbox{Cas}f = \triangle_{S^2}f = \{\hat x^1,\{\hat x^1,f\}\} + \{\hat x^2,\{\hat x^2,f\}\} + \{\hat x^3,\{\hat x^3,f\}\}
\end{equation}
so the $(2l+1)$-dimensional irreducible subrepresentation is the space of spherical harmonics of degree $l$.  Thus $\uu(N)$ converges to $\uu(\infty)$ not only as a Lie algebra, but also as a representation of $\su(2)$, and the irreducible sub-representations can be interpreted as fuzzy spherical harmonics.

In what follows we will need to interpret a number of familiar concepts from linear algebra in terms of fuzzy spheres.  First of all, recall that a matrix $F$ has determinant 0 if and only if it fails to have a multiplicative inverse.  Similarly, a function $f$ on the sphere fails to have a multiplicative inverse if and only if it has a zero somewhere on the sphere.  Therefore, functions on the fuzzy sphere with determinant 0 should be interpreted as functions on the fuzzy sphere with zeros.

This logic can be extended a little by recalling that the spectrum of a matrix is just the set of its eigenvalues:
\begin{equation}
\mbox{Spec}(F) = \{ \lambda\in\CC | F-\lambda \mbox{ is not invertible}\}.
\end{equation}
The analogous object for functions on the sphere is the range:
\begin{equation}
\mbox{Range}(f) = \{ \lambda\in\CC | f-\lambda \mbox{ is not invertible}\}.
\end{equation}
Thus the spectrum of a matrix should be interpreted as the range of a function on the fuzzy sphere.  Support is lent to this assertion by the observations that the spectrum of a hermitian matrix is real, and that
\begin{eqnarray}
\mbox{Spec}(X^i) &=& \left\{\frac{1}{N}-1,\frac{3}{N}-1,\dots,1-\frac{3}{N},1-\frac{1}{N}\right\} \\
&\to& [-1,1] \mbox{ as }N\to\infty.
\end{eqnarray}

The trace is a linear function from $\CC^{N\times N}$ to $\CC$ with the property that the trace of the Lie bracket of any two functions is zero.  Similarly, the integral over $S^2$ with respect to the standard area form $\omega$ is a linear function from the space of functions to $\CC$ which evaluates to zero on the Poisson bracket of any two functions.  So the expression
\begin{equation}
\frac{4\pi}{N}\Tr(F)
\end{equation}
should be interpreted as the integral of $F$ over the fuzzy sphere.  The normalisation has been chosen so that the integral of the identity matrix is $4\pi$.

Given a column vector $v\in\CC^N$, we can form an $N\times N$ matrix $vv^\dagger$.  We would like to know what kind of function this matrix represents.  For simplicity, we will assume that $v$ is an eigenvector of the fuzzy coordinate function $X^3$ with eigenvalue $\lambda$, and we will normalise $v$ so that
\begin{equation}
 \frac{4\pi}{N}\Tr(vv^\dagger) = \frac{4\pi}{N}v^\dagger v = 1.
\end{equation}
Now, $vv^\dagger$ commutes with $X^3$, and since the Hamiltonian vector field associated with $\hat x^3$ is a rotation about the $\hat x^3$-axis, we conclude that $vv^\dagger$ is invariant under rotations about the $\hat x^3$-axis.  Moreover, the integral of the product of $vv^\dagger$ with any power of $X^3$ is
\begin{equation}
\frac{4\pi}{N}\Tr(vv^\dagger (X^3)^m) = \lambda^m.
\end{equation}
The only function on the sphere with analogous properties is a delta function with support along the circle $\hat x^3=\lambda$.  Therefore, $vv^\dagger$ should be interpreted as a delta function with support on a circle.

Morally, one should not expect to find a point-like delta function on the fuzzy sphere or any other non-commutative geometry, because the Heisenberg uncertainty principle forbids precise knowledge of all of the coordinates of a point.  However, a point-like delta function can be approximated by a small circle-like delta function; thus if $v$ is a normalised eigenvector of $X^3$ with eigenvalue $1-1/N$ then $vv^\dagger$ converges in the limit $N\to\infty$ to a point-like delta function with support at the north pole $\hat x^3=1$.

\section{The large $N$ limit of the Nahm transform}
\label{sec:5}

\subsection{The $\uu(N)$ Nahm transform}

The Nahm transform is a bijective correspondence between sets of Nahm data and monopoles (modulo gauge equivalence).  Let $s\in(-v,v)$ and let $T^1(s),T^2(s),T^3(s)$ be three functions of $s$ taking values in the space $\mathfrak{u}(N)$ of $N\times N$ hermitian matrices.
The matrix-valued functions $T^i$ are called $\mathfrak{u}(N)$ \emph{Nahm data} if they satisfy the Nahm equation,
\begin{equation}
\label{NE}
\frac{\dd T^i}{\dd s} = -\ii e\,\epsilon_{ijk}T^jT^k,
\end{equation}
the reality condition,
\begin{equation}
\label{RC}
T^i(-s)=T^i(s)^t,
\end{equation}
and the boundary condition,
\begin{equation}
\label{BC}
T^i(s) = \frac{\ii}{e} \frac{J^i}{v-s} + O(1) \mbox{ as }s\rightarrow v.
\end{equation}

Now let $\T=T^j\otimes\sigma_j$ and $\y=y^j\, \mbox{Id}_N\otimes\sigma_j$, and consider the Weyl equation,
\begin{equation}
\label{WE}
\left(\frac{\dd}{\dd s} + e(\T-\y)\right)\psi = 0
\end{equation}
with boundary condition
\begin{equation}
\label{WBC}
\psi(\pm v,y^j) = 0\, ,
\end{equation}
where $\psi(s,y^j)$ is a $2N$-column vector.  Let us consider what form $\psi$ may take near the boundary $s=v$: suppose that $\psi(s)=f(s)u$, with $f$ a function and $u$ an eigenvector of $\ii J^j\otimes\sigma_j$ with eigenvalue $\lambda$.  Using (\ref{BC}), the differential equation (\ref{WE}) reduces to
\begin{equation}
\frac{\dd f}{\dd s} + \frac{\lambda f}{v-s} = 0.
\end{equation}
The solution of this equation is $f(s)=(v-s)^{\lambda}$, and this decays at $v$ only if $\lambda>0$.  The matrix $\ii J^j\otimes \sigma_j$ has eigenvalues $(\pm N-1)/2$ with eigenspaces of dimension $N\pm1$, and it follows that the space of solutions to (\ref{WE}) satisfying $\psi(v)=0$ is $(N+1)$-dimensional.  Similarly, the space of solutions satisfying $\psi(-v)=0$ is also $(N+1)$-dimensional.  Taking the intersection, the space of solutions to (\ref{WE}), (\ref{WBC}) is 2-dimensional.

To implement the Nahm transform, take a basis $\psi_1,\psi_2$ for the space of solutions to (\ref{WE}), (\ref{WBC}) which is orthonormal:
\begin{equation}
\label{NC}
\int_{-v}^{v} \left( \begin{array}{c} \psi_1^\dagger \\ \psi_2^\dagger \end{array}\right) 
\left(\begin{array}{cc}\psi_1 & \psi_2\end{array}\right) \,\dd s 
= \left(\begin{array}{cc}1&0\\0&1\end{array}\right).
\end{equation}
Then the monopole corresponding to the Nahm data $T^j$ is
\begin{eqnarray}
\label{NT Phi}
\Phi &=& \ii\int_{-v}^{v} \left( \begin{array}{c} \psi_1^\dagger \\ \psi_2^\dagger \end{array}\right) 
\left(\begin{array}{cc}\psi_1 & \psi_2\end{array}\right) s\,\dd s,\\
\label{NT A}
A_j &=& \frac{1}{e}\int_{-v}^{v} 
\left( \begin{array}{c} \psi_1^\dagger \\ \psi_2^\dagger \end{array}\right) 
\frac{\dd }{\dd y^j}\left(\begin{array}{cc}\psi_1 & \psi_2\end{array}\right) \,\dd s .
\end{eqnarray}
It has been proven that $A_j,\Phi$ solve the Bogomolny equation (\ref{bog eq}), and that all solutions of the Bogomolny equation may be determined in this way \cite{cg,hitchin}.

\subsection{The large $N$ limit}

It is easy to see how a set of $\uu(N)$ Nahm data can converge to a set of $\uu(\infty)$ Nahm data.  The matrix-valued functions $T^i(s)$ are functions on the Cartesian product of the fuzzy sphere with the interval $(-v,v)$.  Their large $N$ limit (if it exists) will be a triple of functions $t^i$ on $S^2\times(-v,v)$.  If we let $2\pi N/e\to q$ as in section \ref{sec:2}, then the Nahm equation (\ref{NE}) for $T^i$ converges to the Nahm equation (\ref{NEinfty}) for $t^i$, using the definition (\ref{fuzzy PB}) of the Lie bracket.  Similarly, the boundary condition (\ref{BC}) for $T^i$ converges to the boundary condition (\ref{BCinfty}) for $t^i$.  Thus the restriction of $t^i(s)$ to the interval $[0,v)$ is a set of $\uu(\infty)$ Nahm data.

Note that no information is lost in restricting to the interval $[0,v)$, since the reality condition (\ref{RC}) implies that the Nahm data on the interval $(-v,0]$ are essentially the same as those on $[0,v)$.  However, it should be pointed out that in order to obtain Nahm data for a magnetic bag with non-zero volume (such as the spherical bag), it is necessary that the Nahm data $t^i(s)$ are discontinuous at $s=0$.  It is of course perfectly possible for a sequence of continuous functions to have a discontinuous limit.

To understand how the $\uu(N)$ Nahm transform converges to the $\uu(\infty)$ Nahm transform described in theorem \ref{thm:1} requires more effort; the two Nahm transforms are quite different in character.  The starting point is the \emph{spectral index}, which we define to be half of the difference between the dimensions of the positive and negative eigenspaces of the matrix $\T(s)-\y$.  We have already seen that the space of solutions to (\ref{WE}), (\ref{WBC}) is 2-dimensional precisely because the spectral index increases from $-1$ at $-v$ to $+1$ at $v$ (this is actually a simple finite-dimensional example of what mathematicians call a ``spectral flow'').  The following proposition shows that the spectral index is a strictly increasing function of $s$:
\begin{prop}
\label{prop:1}
Suppose that $T^j(s)$ solve the Nahm equation (\ref{NE}) and that $\T(s)-\y$ has a 0-eigenvalue at some value $s_0$ of $s$ and is non-singular on the intervals $(s_0-\epsilon,s_0)$ and $(s_0,s_0+\epsilon)$.  Then the spectral index takes constant values $I_+,I_-$ on the intervals $(s_0,s_0+\epsilon)$, $(s_0-\epsilon,s_0)$, and the difference $I_+-I_-$ is positive and equal to the dimension of the kernel of $\T(s_0)-\y$.
\end{prop}
\begin{proof}
We begin by rewriting the Nahm equation (\ref{NE}) as follows:
\begin{equation}
\frac{\dd}{\dd s} (\T-\y) = -e\Im( (\T-\y)^2 ).
\end{equation}
Here we are identifying the $2\times 2$ matrices $\mbox{Id}_2,-\ii\sigma^j$ with quaternions, and $\Im$ denotes the imaginary quartenionic part (so that $\Im\sigma^j=\sigma^j$ and $\Im\mbox{Id}_2=0$).  It follows that, in a neighbourhood of $s_0$,
\begin{equation}
\label{eq:5.1}
\T(s)-\y = (\T_0-\y) - (s-s_0)e\Im(\T_0-\y)^2 + O((s-s_0)^2),
\end{equation}
where $\T_0=\T(s_0)$.

Now suppose that $u$ is a 0-eigenvector of $\T_0-\y$.  Then
\begin{eqnarray}
0 &=& u^\dagger (\T_0-\y)^2 u \\
\label{eq:5.2}
&=& u^\dagger \Re (\T_0-\y)^2 u + u^\dagger \Im (\T_0-\y)^2 u
\end{eqnarray}
(where $\Re$ denotes the real quaternionic part).  Now $\Re(\T_0-\y)^2=(T_0^i-y^i)(T_0^i-y^i)$ is a positive definite hermitian matrix, and it follows that $-\Im(\T_0-\y)^2$ is positive definite on the kernel of $\T_0-\y$.  This information is sufficient to guarantee that the spectral index jumps by $n$ at $s_0$, where $n$ is the dimension of the kernel of $\T_0-\y$.

To see this, note that equation (\ref{eq:5.1}) takes the form,
\begin{equation}
\T(s)-\y = \left(\begin{array}{c|c} A+(s-s_0)C & (s-s_0)D^\dagger \\ \hline (s-s_0)D & (s-s_0)B \end{array}\right) + O((s-s_0)^2),
\end{equation}
with $A$ denoting the restricition of $\T_0-\y$ to the orthogonal complement of its kernel, $B$ denoting positive definite restriction of $-e\Im(\T_0-\y)^2$ to the kernel of $\T_0-\y$, and so on.  Using a matrix identity, the characteristic polynomial is
\begin{eqnarray}
\nonumber
\det(\T(s)-\y-\lambda) &=& \det(-\lambda+A+O(s-s_0)) \\
&& \times \det(-\lambda + (s-s_0)B + O((s-s_0)^2)).
\end{eqnarray}
For small enough $s-s_0$, the signs of the roots of the first factor on the right are independent of $s-s_0$.  The second factor on the right has only positive roots when $s>s_0$, and only negative roots when $s<s_0$, because $B$ is positive definite.  So, for small enough $\delta>0$, the difference between the spectral indices at $s_0+\delta$ and $s_0-\delta$ is equal to the dimension of the kernel of $\T_0-\y$.

\end{proof}

An immediate corollary is the following:
\begin{cor}
Suppose that $T^j(s)$ are a set of Nahm data.  Then $\T(s)-\y$ is singular precisely at the values $\pm s_0$ of $s$, for some $0\leq s_0<v$.
\end{cor}
\begin{proof}
We have already seen that the boundary conditions (\ref{BC}) (together with the reality condition (\ref{RC})) imply that the spectral index is $\pm1$ as $s\to\pm v$.  The previous proposition says that the spectral index increases by at least 1 at every point where $\T(s)-\y$ is singular, so $\T(s)-\y$ can be singular at either 1 or 2 values of $s$.  The reality condition (\ref{RC}) implies that in the former case the points must be $s=0$, and in the latter case they must be $s=\pm s_0$ with $0<s_0<v$.
\end{proof}

Now we shall consider what happens to $s_0$ in the limit $N\to\infty$.  The restriction $\T(s)-\y$ to the interval $I=[0,v)$ is a $2\times2$ matrix-valued function on the Cartesian product of fuzzy sphere with $I$, and $s_0$ is by definition the unique point where this function is not invertible.  The large $N$ limit (if it exists) is a matrix-valued function $(t^i(s)-y^i)\sigma_i$ on $S^2\times I$, and following the logic introduced in the previous section, the large $N$ limit of $s_0$ is the point where this function is not invertible.  The function $(t^i(s)-y^i)\sigma_i$ fails to be invertible precisely when $t^i(s)=y^i$ at some point $p$ on the sphere.  
So $s_0(y^i)$ is defined by the property that $y^i$ is contained in the image of $t^i(\cdot,s_0)$.  This definition of $s_0(y^i)$ agrees precisely with the the definition of the function $\phi(y^i)$ obtained under the $\uu(\infty)$ Nahm transform in theorem \ref{thm:1}.

Thus the $\uu(\infty)$ Nahm transform will be obtained from the $\uu(N)$ Nahm transform, provided that the value of $\|\Phi\|$ obtained from (\ref{NT Phi}) converges to $s_0$ as $N\to\infty$.  From equation (\ref{NT Phi}), this would be the case if the two solutions $\psi_1,\psi_2$ of the Weyl equation became localised around the points $\pm s_0$ as $N\to\infty$.  We conjecture that this is indeed the case, and in the following subsection we will provide some arguments in favour of this conjecture.

\subsection{Localisation of solutions}

It is not unreasonable to expect the solutions of the system (\ref{WE}), (\ref{WBC}) to localise about the points $\pm s_0$.  As we have already seen, the existence of points $\pm s_0$ where the spectral index of $\T(s)-\y$ jumps is what allows this system to have solutions.

If $r$ is large, then $s_0$ is close to $v$, and we can approximate Nahm data using the boundary condition (\ref{BC}):
\begin{equation}
\label{approx ND}
T^i(s) \approx \frac{\ii}{e}\frac{J^i}{v-s}.
\end{equation}
The Weyl equation associated to this approximate Nahm data can be solved exactly.  The matrix $\y$ has eigenvalues $\pm r$ with eigenspaces of dimension $N$ and the matrix $\T(s)$ has eigenvalues $(\pm N-1)/(2e(v-s))$ with eigenspaces of dimension $N\pm 1$.  Let $u$ be a vector in the 1-dimensional intersection of the eigenspaces with eigenvalues $r$ and $(N-1)/(2e(v-s))$ respectively.  We make an ansatz $\psi(s)=f(s)u$: then the Weyl equation (\ref{WE}) becomes
\begin{equation}
\left(\frac{\dd }{\dd s} + \frac{N-1}{2}\frac{1}{v-s} - er\right)f = 0.
\end{equation}
The unique solution of this equation that decays at $s=v$ and $s=-\infty$ is
\begin{equation}
f(s) = \exp(ers)(v-s)^{(N-1)/2}.
\end{equation}
The solution $f$ achieves a maximum at $s_0 = v - (N-1)/(2er)$.  One also finds that $f''(s_0)/f(s_0)=-e^2r^2/(N-1)$ diverges as $N,e\to\infty$, indicating a localisation of the solution.  An approximate solution of the Weyl equation localised near $s=-s_0$ can be constructed in a similar way.  So for large $r$, the Weyl equation has two approximate solutions which localise at $s=\pm s_0$.

One interesting feature of the approximate solution constructed above is that it passes through the kernel of $\T(s)-\y$ at $s_0$.  Motivated by this, we now consider solutions of the Weyl equation (\ref{WE}) for which $\psi_0:=\psi(s_0)$ is a 0-eigenvector of $\T_0-\y:=\T(s_0)-\y$.  For such solutions, the quantity
\begin{equation}
\lambda := \frac{1}{\psi^\dagger\psi}\frac{\dd^2 (\psi^\dagger\psi)}{\dd s^2}(s_0)
\end{equation}
measures the localisation: if $\lambda$ is large and negative, the solution is localised.  Elementary analysis of the differential equations (\ref{NE}), (\ref{WE}) shows that
\begin{equation}
\lambda = \frac{-2e}{\psi_0^\dagger \psi_0}\psi_0^\dagger\frac{\dd \T}{\dd s}(s_0)\psi_0 = 2e^2\, \frac{\psi_0^\dagger \Im(\T_0-\y)^2 \psi_0}{\psi_0^\dagger \psi_0}.
\end{equation}
From the discussion around equation (\ref{eq:5.2}), we know that the quantity on the right is negative, as required.  Showing that this quantity tends to $-\infty$ is more subtle, and requires us to make sense of $\psi_0\psi_0^\dagger$ in the large $N$ limit.

For convenience we will denote $t_0^i=t^i(s_0)$.  The $2N\times 2N$ hermitian matrix $\psi_0\psi_0^\dagger$ converges to a $2\times 2$ hermitian matrix-valued function $u$ on the sphere as $N\to\infty$.  Furthermore, since
\begin{equation}
0 = \lim_{N\to\infty} (\T_0-\y)^2\psi_0\psi_0^\dagger = (t^i_0-y^i)(t^i_0-y^i)u,
\end{equation}
it is clear that $u$ is a delta function centred on the point $p$ where $t^i_0(p)=y^i$.  We will write $u=(w^0\mbox{Id}_2+w^j\sigma_j)\delta_p$ for real numbers $w^0,w^i$.

For any function $F$ on the fuzzy sphere whose large $N$ limit is a function $f$ on $S^2$,
\begin{eqnarray}
0 &=& \lim_{N\to\infty} \frac{4\pi}{2N} \Tr_2\Tr_N\left( \{\T_0,F\}uu^\dagger \right) \\
&=& \int_{S^2} \{t_0^i,f\} u^i\,\omega \\
&=& \{t_0^i,f\}(p)w^i.
\end{eqnarray}
Now the vector field $\{t_0^i,f\}$ is tangent to $t^i_0(S^2)$, and in particular all tangent vectors to the point $p$ can be obtained by making appropriate choices for $f$.  It follows that $w^i$ is a normal vector to the sphere at $p$.

The matrix $\psi\psi^\dagger$ has the property that it is equal to a constant times its square.  The same property should hold for the function $u$, and this implies that
\begin{eqnarray}
\label{eq:5.3}
w^0w^0 + w^iw^i &=& c w^0 \\
\label{eq:5.4}
2w^0w^i &=& c w^i
\end{eqnarray}
for some constant $c$.  This pair of equations has only two non-zero solutions: either $w^iw^i=w^0w^0$, or $w^i=0$ and $w^0\neq0$.  For the approximate large-$r$ Nahm data (\ref{approx ND}) one can check directly that $w^iw^i=w^0w^0$, so by continuity we expect this to be the case everywhere.  Henceforth we assume that $w^iw^i=w^0w^0$.

Now, the large $N$ limit of $\lambda/2e$ is
\begin{eqnarray}
\lim_{N\to\infty} \frac{\lambda}{2e} &=& \left( \int_{S^2} \frac{1}{2}\Tr_2 (u) \omega \right)^{-1} \int_{S^2} \frac{1}{2}\Tr_2\left(u\frac{\dd t^i}{\dd s}(s_0)\sigma_i\right)\omega  \\
&=& \frac{w^i}{w^0}\frac{\dd t^i}{\dd s}(s_0,p)
\end{eqnarray}
The quantity on the second line of this equation is equal to $-|\dd t^i/\dd s(s_0,p)|$, because the Nahm equation (\ref{NEinfty}) implies that $\dd t^i/\dd s$ is perpendicular to $t^i(S^2)$ (and we know that $\lambda$ is negative).  Therefore $\lambda\to-\infty$ as $N\to\infty$, and solutions of the Weyl equation (\ref{WE}) which pass through the kernel of $\T(s)-\y$ at $s_0$ become localised in the limit $N\to\infty$.  Such solutions are expected to be good approximations to the solutions of the system (\ref{WE}), (\ref{WBC}) in the large $N$ limit.


\subsection{The brane picture}

\begin{figure}
\begin{center}
\epsfig{file = 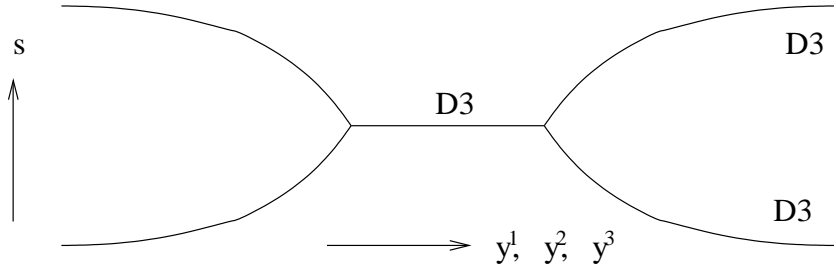, scale = 1.0}
\end{center}
\caption{A magnetic bag as a configuration of intersecting D3-branes.}
\label{fig:0}
\end{figure}

In type IIB string theory, monopoles with topological charge $N$ are represented by $N$ D1-branes stretched between 2 parallel $D3$-branes.  In the effective gauge theory on the D3-brane the D1-branes appear as magnetic monopoles \cite{douglas&li}, while in the effective theory on the D1-brane the D3-branes are described by solutions of the Nahm equation \cite{diaconescu}.  So from the point of view of string theory, the Nahm transform for charge $N$ monopoles relates two different descriptions of the same physical entity.

In the $N\to\infty$ limit, the D1-branes merge and the resulting configuration consists of two warped D3-branes which meet and intersect a third D3-brane, as sketched in figure \ref{fig:0} \cite{bol1}.  From the D3-brane perspective, the shapes of these two D3-branes are described by graphs of real functions on $\RR^3$, and these functions are nothing other than the scalar field $\phi$ associated with the magnetic bag and its negative $-\phi$.  One can also choose to parametrise each D3-brane from the D1-brane perspective, using a function $t^i:S^2\times I\to\RR^3$.  The relation between these two parametrisations is precisely the Nahm transform described in theorem \ref{thm:1}, so once again the Nahm transform relates two descriptions of the same physical system.

Our description of the Nahm transform at large $N$ extends the picture presented in \cite{cmt}: while the analysis in \cite{cmt} was restricted to ``fuzzy funnel'' configurations in the $v\to\infty$ limit with spherical symmetry, here we have been able to give a complete account of the large $N$ Nahm transform for arbitrary configurations with finite $v$.  Reference \cite{cmt} also analyses non-BPS aspects of the D1-D3 system using DBI actions: it would be interesting to extend this analysis to the more general setting considered here.


\section{An example: the magnetic disc}
\label{sec:6}

In the preceding section we described how the $\uu(N)$ Nahm transform converges to the $\uu(\infty)$ Nahm transform as $N\to\infty$.  In this section we will present a simple example of this limiting process, which fully supports the picture of convergence described above.

\begin{figure}[tb]
\epsfig{file = 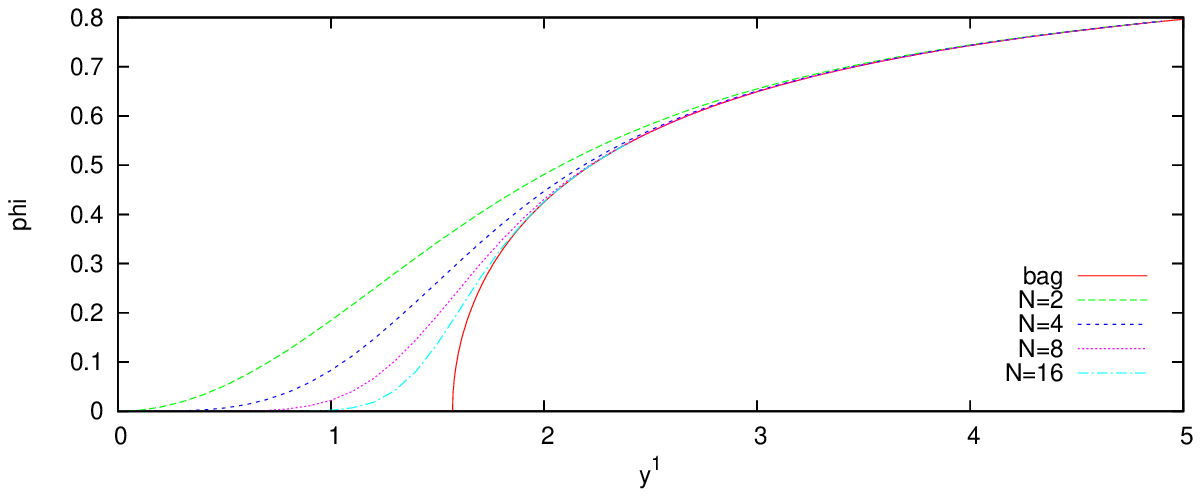, scale = 1}
\epsfig{file = 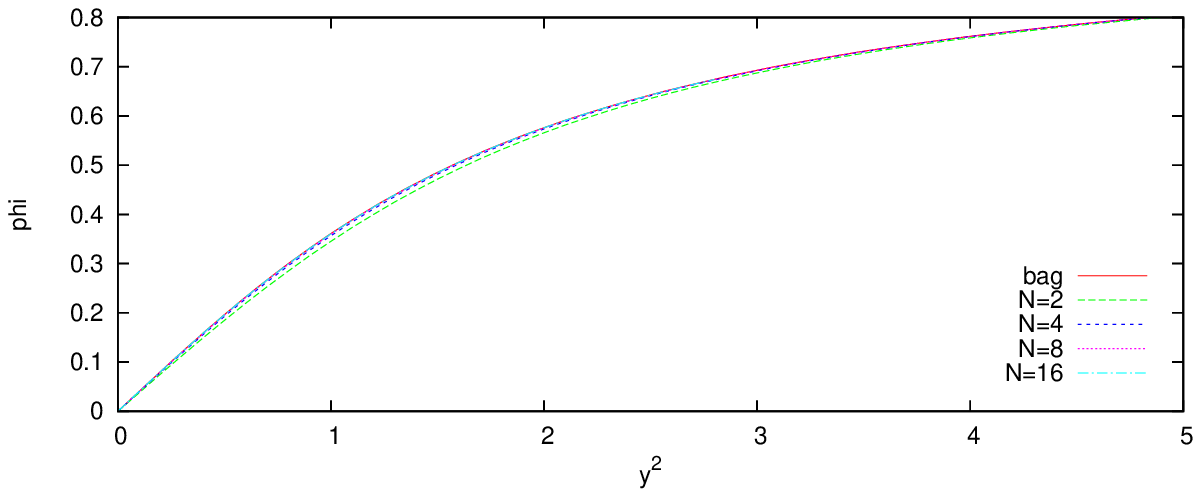, scale = 1}
\epsfig{file = 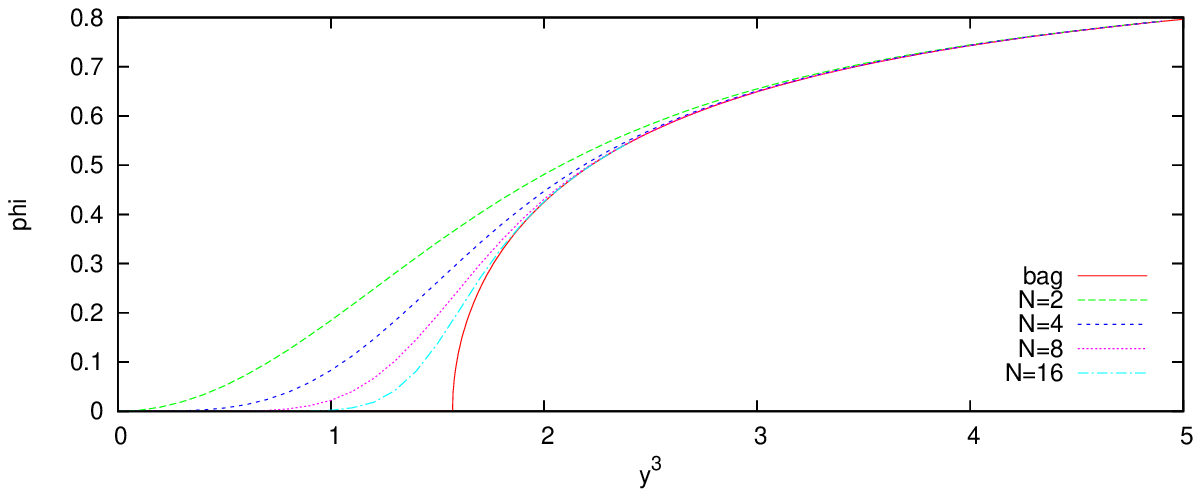, scale = 1}
\caption{Graphs comparing the $k^2=0$ monopoles with a range of values of $N$.  The top, middle and bottom graphs show $\|\Phi\|$ along the positive $y^1$-, $y^2$- and $y^3$-axes.  Charges $N=2$, 4, 8 and 16 are indicated by long-dashed, short-dashed, dotted and dash-dotted lines, while the bag limit is indicated by a solid line.  The graphs show that the function $\|\Phi\|$ converges to the bag limit $\phi$.}
\label{fig:1a}
\end{figure}

\begin{figure}[tb]
\epsfig{file = 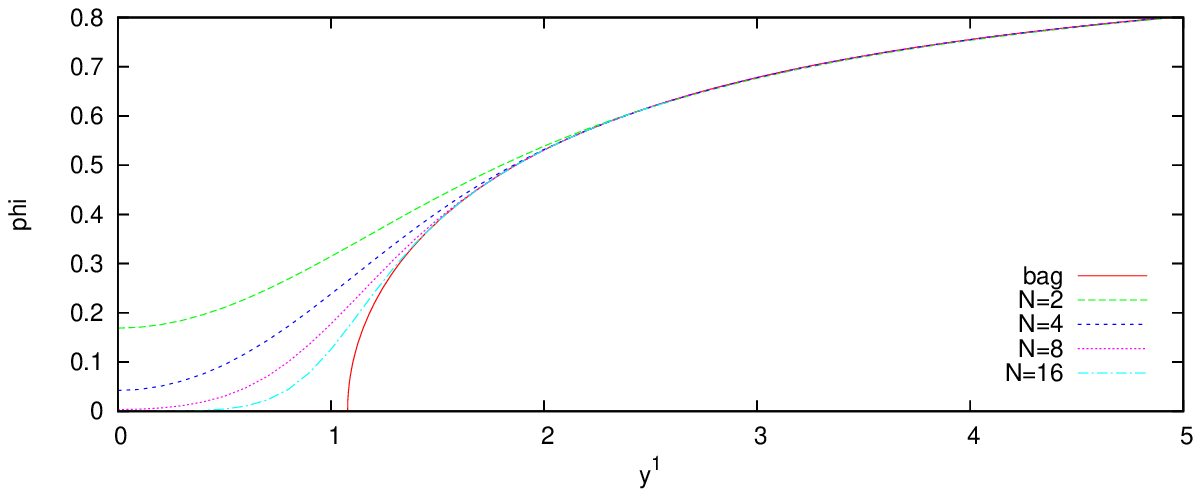, scale = 1}
\epsfig{file = 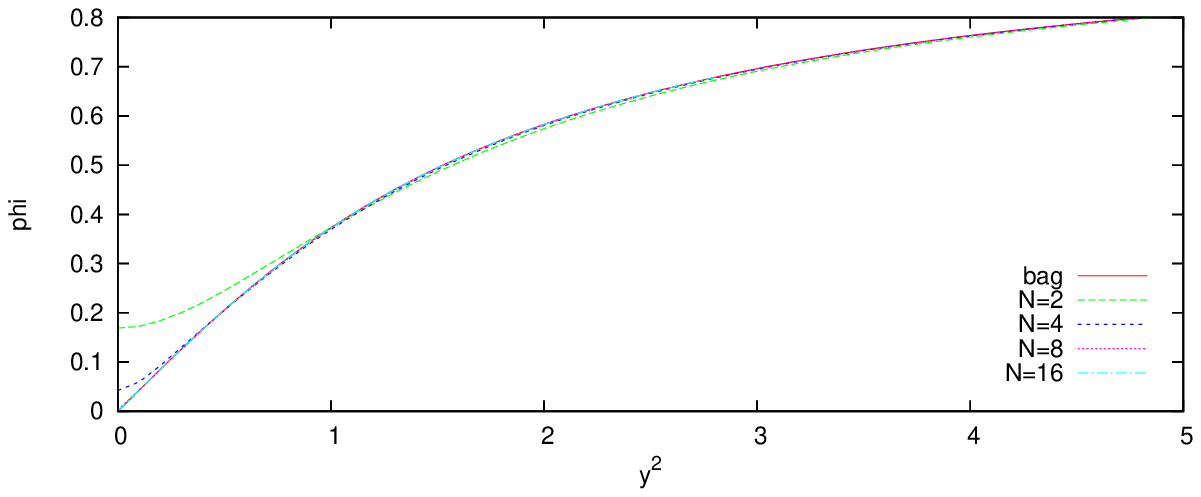, scale = 1}
\epsfig{file = 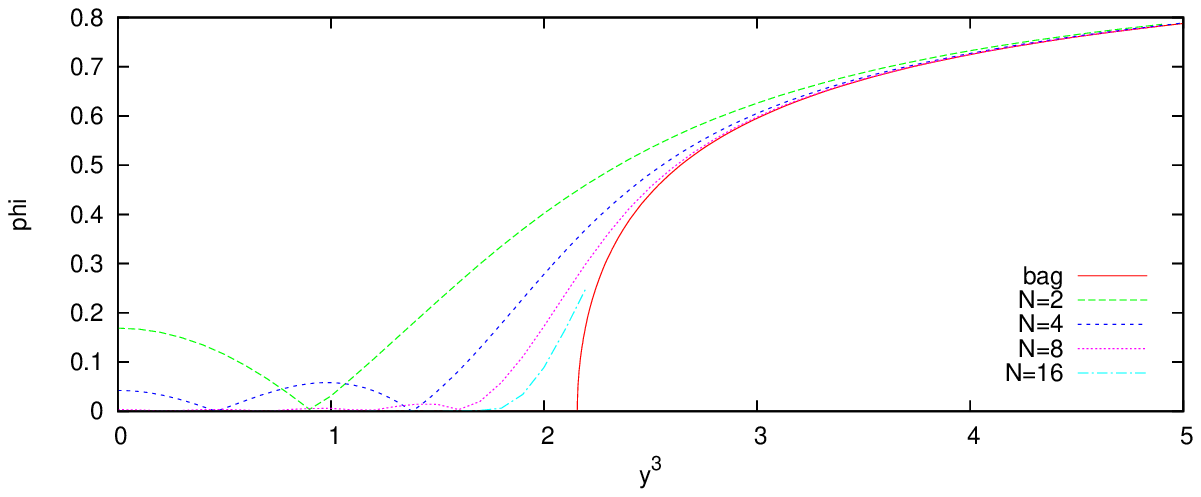, scale = 1}
\caption{Graphs comparing the $k^2=3/4$ monopoles with a range of values of $N$.  The top, middle and bottom graphs show $\|\Phi\|$ along the positive $y^1$-, $y^2$- and $y^3$-axes.  Charges $N=2$, 4, 8 and 16 are indicated by long-dashed, short-dashed, dotted and dash-dotted lines, while the bag limit is indicated by a solid line.  The graphs show that the function $\|\Phi\|$ converges to the bag limit $\phi$.}
\label{fig:1b}
\end{figure}

Although explicit examples of $\uu(N)$ Nahm data are more common than explicit examples of monopoles, they are still relatively rare.  In fact, there is only one known example of a family of Nahm data which allows one to take an $N\to\infty$ limit.  The family of Nahm data in question was written down by Ercolani and Sinha \cite{es}, and takes the following form:
\begin{equation}
\label{ES ND}
T^j(s) = \frac{\ii}{e}f_j(s)J^j.
\end{equation}
Here the functions $f_j(s)$ are given in terms of Jacobi elliptic functions with elliptic modulus $k$, complementary elliptic modulus $k'=\sqrt{1-k^2}$, and complete elliptic integral of the first kind $K$:
\begin{equation}
f_1(s) = \frac{Kk'}{ v}\mbox{nc}\left(\frac{Ks}{ v}\right),\; f_2(s) = \frac{Kk'}{ v}\mbox{sc}\left(\frac{Ks}{ v}\right),\; f_3(s) = \frac{K}{ v}\mbox{dc}\left(\frac{Ks}{ v}\right).
\end{equation}

The simple form of this Nahm data makes it easy to take the $N\rightarrow\infty$ limit:
\begin{equation}
t^j = \lim_{N\rightarrow\infty} T^j = \frac{q}{4\pi} f_j(s)\hat x^j.
\end{equation}
The magnetic bag obtained from this $\uu(\infty)$ Nahm data is degenerate, in the sense that the volume contained inside the bag surface $\Sigma$ is 0.  The surface $\Sigma$ takes the form of an ellipse lying in the $y^1,y^3$-plane with eccentricity $k$:
\begin{equation}
\Sigma = \Bigg\{(y^1,y^2,y^3)\bigg|y^2=0, (y^1)^2+(k'y^3)^2\leq\left(\frac{qk'K}{4\pi v}\right)^2 \Bigg\}.
\end{equation}
This flat magnetic bag will be called a ``magnetic disc''.  In the case $k=0$ the monopoles associated with (\ref{ES ND}) develop an axial symmetry and the ellipse becomes a circle; these monopoles and their large $N$ limit were analysed by Bolognesi \cite{bol1}.  A different large $N$ limit of the Nahm data (\ref{ES ND}) was discussed in \cite{hw}.

It is possible to perform the Nahm transform for the Nahm data (\ref{ES ND}) numerically \cite{dk}, and this provides a neat illustration of the limiting procedure described in the previous section.  Our numerical implementation of the Nahm transform follows the method described in \cite{hs}: the $N+1$ solutions of the Weyl equation (\ref{WE}) satisfying the boundary condition $\psi(-v)=0$ are constructed by shooting from $s=-v$, and the $N+1$ solutions satisfying $\psi(v)=0$ are constructed in an analogous manner.  The 2-dimensional space of solutions is found by matching these two sets of solutions up at $s=0$, which amounts to finding the kernel of a $2N\times 2(N+1)$ matrix.  The only deviation from \cite{hs} is that we find this kernel using singular value decomoposition rather than row reduction with back substitution, the former being more accurate for large matrices.

\begin{figure}[tb]
\epsfig{file = 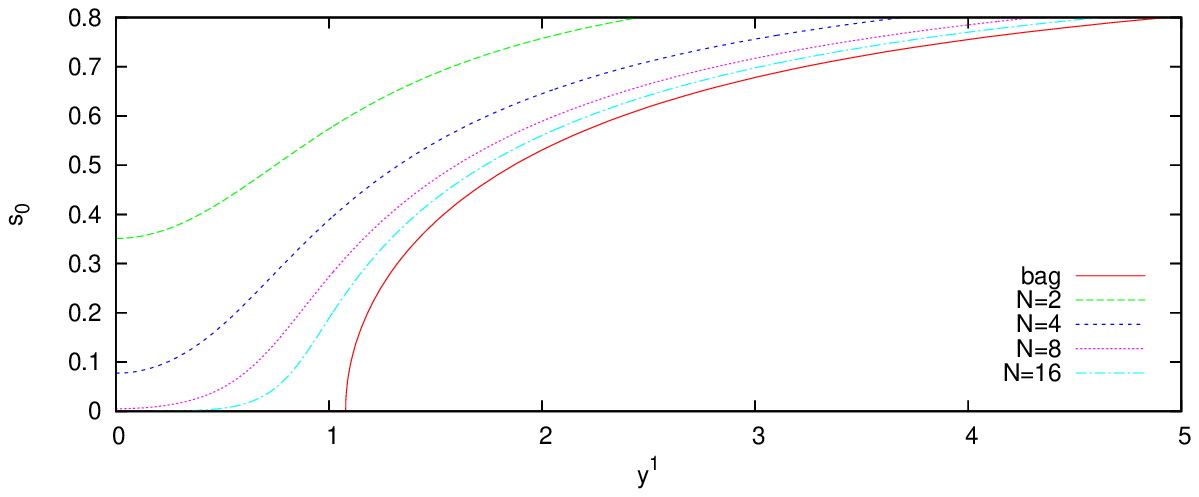, scale = 1}
\epsfig{file = 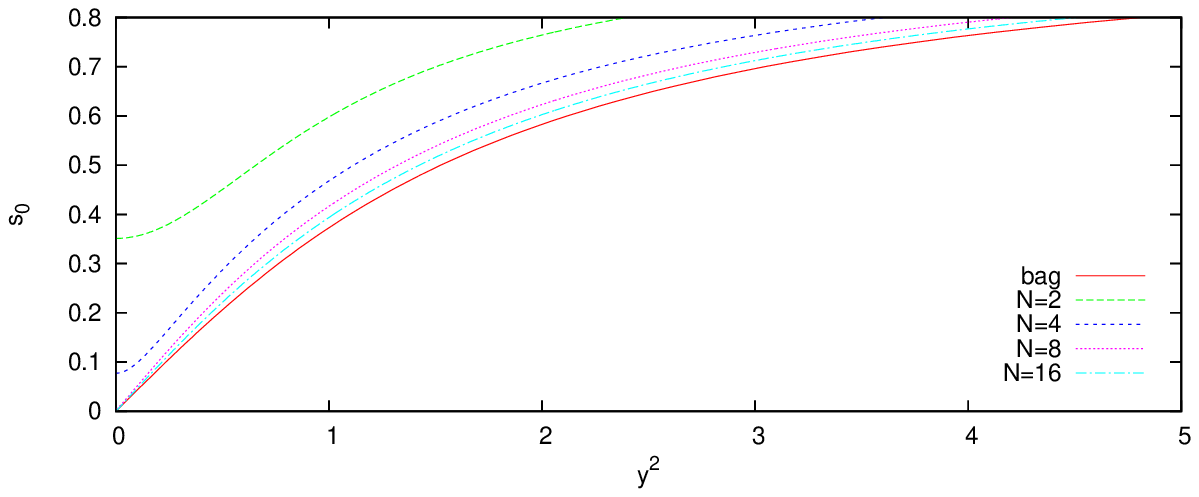, scale = 1}
\epsfig{file = 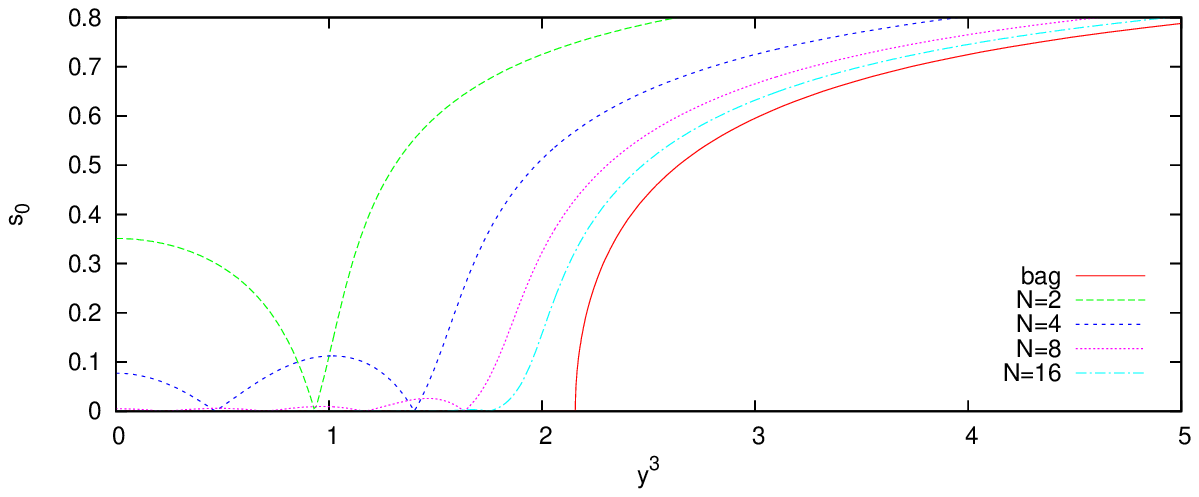, scale = 1}
\caption{Graphs comparing $s_0$ for the $k^2=3/4$ monopoles with a range of values of $N$.  The top, middle and bottom graphs show $s_0$ along the positive $y^1$-, $y^2$- and $y^3$-axes.  Charges $N=2$, 4, 8 and 16 are indicated by long-dashed, short-dashed, dotted and dash-dotted lines, while the bag limit is indicated by a solid line.  The graphs show that the function $s_0 $ converges to the bag limit.}
\label{fig:2}
\end{figure}

Our results are illustrated in figures \ref{fig:1a} and \ref{fig:1b}.  These graphs show the value of $\|\Phi\|$ along the 3 coordinate axes for two different values of elliptic parameter $k$.  In each plot the solid line is the limiting value of $\phi$ predicted by the $\uu(\infty)$ Nahm transform, and the other lines correspond to different values of the topological charge $N$ up to $N=16$.  The graphs clearly show that the monopoles converge to the magnetic disc as expected.

For values of $N$ greater than 16, the output of our numerical Nahm transform is unreliable except in a region close to the monopole core.  The source of error is the singular value decomposition: at points far from the monopole core, the matrix whose kernel needs to be found has very large eigenvalues, and this makes it difficult to find the kernel with sufficient accuracy.  However, as figures \ref{fig:1a} and \ref{fig:1b} show, already at charge 16 the value of $\|\Phi\|$ away from the monopole core is well-approximated by the $\uu(\infty)$ Nahm transform.  This suggests that the $\uu(\infty)$ Nahm transform could be used to study large-charge monopoles in regions where the $\uu(N)$ Nahm transform is difficult to implement.

Our analysis in the previous section relied on the point $s_0$ at which the spectral index jumps being a good approximation to $\phi$, and it is rewarding to verify this for the Nahm data under consideration.  First of all, in figure \ref{fig:2} we have plotted $s_0$ as a function of $y^i$ in a similar manner to figure \ref{fig:1b}.  These plots demonstrate that $s_0$ converges to the bag limit $\phi$ as expected.  In figure \ref{fig:3} we have plotted the values of $s_0$ and $\|\Phi\|$ for a charge 4 monopole with $k^2=3/4$, and the bag limit $\phi$, along the $y^3$-axis.  It is notable that $\|\Phi\|$ is much closer to the bag limit than $s_0$; so this family of monopoles converges to the magnetic disc much faster than the arguments presented in section \ref{sec:5} suggest.  It is also intriguing that the zeros of $s_0$ and $\|\Phi\|$ agree almost exactly.  If this was true in general it would provide a useful approximation to the zeros of $\Phi$, since $s_0$ is much easier to calculate from the Nahm data than $\Phi$.

\begin{figure}
 \epsfig{file = 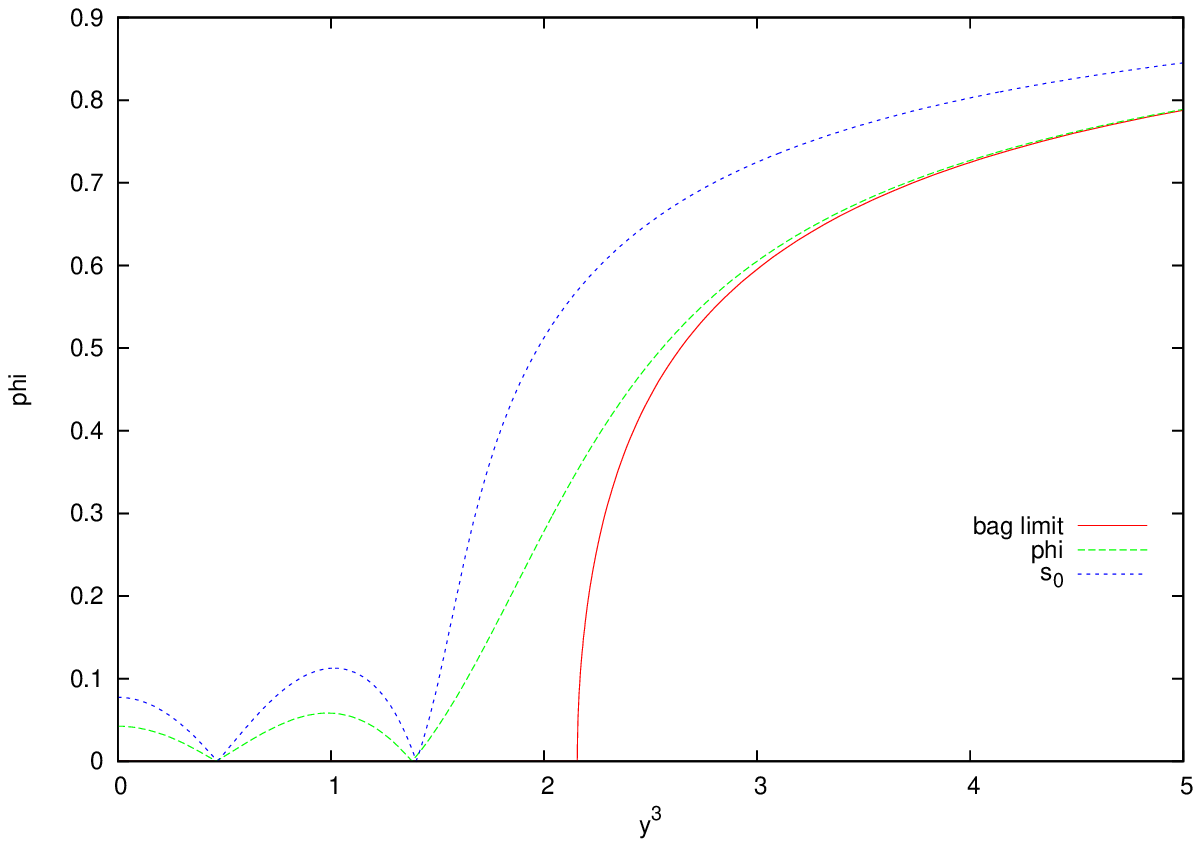, scale = 1.0}
 \caption{Graph comparing $\|\Phi\|$ and $s_0$ for an $N=4$, $k^2=3/4$ monopole and the bag limit $\phi$ along the $y^3$-axis.  Notice that $\|\Phi\|$ is much closer to the bag limit than $s_0$, and also that zeros of $s_0$ and $\|\Phi\|$ almost coincide.}
\label{fig:3}
\end{figure}

\section{Conclusion}
\label{sec:7}

We have described a simple transform which relates magnetic bags to $\uu(\infty)$ Nahm data.  We have argued using fuzzy spheres that this Nahm transform is the large $N$ limit of the Nahm transform relating charge $N$ monopoles with $\uu(N)$ Nahm data.  Finally, we have presented numerical evidence that this is the case, and that the $\uu(\infty)$ Nahm transform approximates monopoles well even at relatively low charge.

We hope that the $\uu(\infty)$ Nahm transform represents a significant step towards understanding magnetic bags.  At present, understanding of magnetic bags is limited by a lack of examples: for instance, no examples of monopoles with topological charge greater than 7 have been shown to  resemble the spherical bag (although the icosahedral $N=11$ monopole is conjectured to \cite{lw}).  We saw in section \ref{sec:3} that the spherical magnetic bag has vanishing conserved charges $C_l$, so a promising way to construct such monopoles would be to look for $\uu(N)$ Nahm data whose conserved charges are in some sense small.

Ideally, one would like to have an analytic proof of the magnetic bag conjecture.  It may be easier first to consider the analogous conjecture for Nahm data: that every set of $\uu(\infty)$ Nahm data is the limit of a sequence of $\uu(N)$ Nahm data.  Proving this statement would presumably require a firmer analytical understanding of fuzzy spheres than has been presented here, and reference \cite{fuzzy} may help in this respect.

There are a number of questions raised by the present work which seem worthy of further investigation.  First of all, the moduli spaces of charge $N$ monopoles and of $\uu(N)$ Nahm data are known to be hyperk\"ahler manifolds, and the $\uu(N)$ Nahm transform is an isometry between these two manifolds.  It would be interesting to investigate whether similar statements hold for the $\uu(\infty)$ Nahm transform.

Second, since it is known that harmonic functions on $\RR^3$ can be used to construct 4-dimensional hyperk\"ahler metrics via the Gibbons-Hawking ansatz, our Nahm transform associates hyperk\"ahler metrics to solutions of the $\uu(\infty)$ Nahm equation.  It is also known that hyperk\"ahler metrics can be constructed from solutions to the Nahm equation associated with the Lie algebra of divergence-free vector fields on a 3-manifold \cite{ajs}.  Now $\uu(\infty)$ is the Lie algebra of divergence-free vector fields on $S^2$, so it seems plausible that our Nahm transform is related to the results of \cite{ajs}.

Third, although we have worked exclusively with $S^2$, one can associate a $\uu(\infty)$ Nahm equation to any surface equipped with an area form.  For example, if $\RR^2$ is equipped with its standard area form $\omega=\dd x^1\wedge\dd x^2$ then the associated $\uu(\infty)$ Nahm equation has the following simple solution \cite{ward}:
\begin{equation}
t^1(s) = x^1,\quad t^2(s)=x^2,\quad t^3(s)=\frac{4\pi s}{q}.
\end{equation}
If we take the range of $s$ to be $[0,\infty)$, then application of the $\uu(\infty)$ Nahm transform yields the following configuration on $\RR^3$:
\begin{equation}
\label{wall nahm data}
\phi(y^i) = \begin{cases} \frac{q}{4\pi}y^3 & y^3\geq0 \\ 0 & y^3<0 \end{cases}.
\end{equation}
This coincides with Lee's approximation to a monopole wall \cite{lee}.  It seems likely that this configuration is a limit of a non-abelian monopole wall, and that this version of the $\uu(\infty)$ Nahm transform could be derived as a limit of the Nahm transform for monopole walls.  Note however that the Nahm transform for monopole walls is only partially understood \cite{ward05}.

Fourth, although we formulated the $\uu(\infty)$ Nahm transform as a construction for harmonic functions on $\RR^3$, the construction seems to work equally well for other 3-manifolds.  Consider for example the ball model $\mathbb{H}^3$ of hyperbolic space, with coordinates $y^i$ satisfying $y^iy^i\leq1$ and metric
\begin{equation}
g = \left(\frac{2}{1-y^jy^j}\right)^2\dd y^i\dd y^i.
\end{equation}
By following the steps in the proof of theorem \ref{thm:1}, a harmonic function on $\mathbb{H}^3$ can be associated with a solution of the following Nahm equation:
\begin{equation}
\label{NEhyp}
 \left(\frac{1-t^jt^j}{2}\right) \frac{\dd t^i}{\dd s} = \frac{4\pi}{q}\frac{1}{2}\epsilon_{ijk} \{t^j,t^k\}.
\end{equation}
A Nahm transform for monopoles on hyperbolic space is known only in the case when the product of the Higgs vacuum expectation value with the scalar curvature is an integer \cite{ward99}, and it is an open problem to determine whether or not a Nahm transform exists in the general case.  Our results suggest that if a Nahm transform exists, it ought to reduce to (\ref{NEhyp}) in the large $N$ limit.  So it would be interesting to investigate $\uu(N)$ generalisations of this equation.

\subsubsection*{Acknowledgements}
I am grateful to Richard Ward for reading an early version of this manuscript and for pointing out the connection between equation (\ref{wall nahm data}) and reference \cite{lee}.  I acknowledge funding from the EPSRC under grant number EP/G038775/1.

\end{document}